\documentclass[conference,letterpaper]{IEEEtran}
\IEEEoverridecommandlockouts
\usepackage[hidelinks]{hyperref}
\usepackage[utf8]{inputenc} 
\usepackage[T1]{fontenc}
\usepackage{url}
\usepackage{ifthen}
\newboolean{isarxiv}
\setboolean{isarxiv}{true}
\newcommand{\arxivonly}[1]{\ifthenelse{\boolean{isarxiv}}{#1}{}}
\newcommand{\confonly}[1]{\ifthenelse{\boolean{isarxiv}}{}{#1}}
\usepackage{cite}
\usepackage[cmex10]{amsmath}
\usepackage{siunitx-v2}
\usepackage{amsfonts}
\usepackage{amssymb}
\usepackage{amsthm}
\usepackage{bm}
\usepackage{algorithmic}
\usepackage{graphicx}
\usepackage{textcomp}
\usepackage{xcolor}
\definecolor{airforceblue}{rgb}{0.36,
0.54, 0.66}
\definecolor{amber}{rgb}{1.0, 0.75,
0.0}
\definecolor{amethyst}{rgb}{0.6, 0.4,
0.8}
\definecolor{applegreen}{rgb}{0.55,
0.71, 0.0}
\definecolor{antiquebrass}{rgb}{0.8,
0.58, 0.46}
\usepackage{enumitem}
\usepackage{mathtools}
\PassOptionsToPackage{hyphens}{url}
\usepackage{tikz,pgfplots}
\usetikzlibrary{positioning}
\usetikzlibrary{decorations.pathreplacing}
\usetikzlibrary {arrows.meta,bending,positioning}
\pgfdeclareradialshading{sphere}{\pgfpoint{0cm}{0cm}}% 
{rgb(0cm)=(1,1,1);
rgb(0.7cm)=(0.7,0.1,0); rgb(1cm)=(0.5,0.05,0); rgb(1.05cm)=(1,1,1)}
\usepackage{grffile}
\pgfplotsset{compat=newest}
\usepgfplotslibrary{patchplots}
\usetikzlibrary{backgrounds,plotmarks,fit,positioning,arrows,shapes,shapes.multipart,calc,arrows.meta}
\usetikzlibrary{fit,positioning,arrows,shapes,shapes.multipart,calc,arrows.meta,shapes.geometric,shapes.misc}
\usetikzlibrary{decorations.pathreplacing,angles,quotes,calligraphy}
\usetikzlibrary{automata}

\allowdisplaybreaks

\usepackage{subfigure}
\usepackage{vmr-symbols-vecbold}
\usepackage{standard-macros}

\usepackage{soul} % For highlighting
\sethlcolor{yellow} % Default highlight color
\newtheorem{definition}{Definition}
\newtheorem{lemma}{Lemma}

\usepackage[toc,acronym,nonumberlist,nopostdot]{glossaries}
\newacronym{IoT}{IoT}{Internet of Things}
\newacronym{PUPE}{PUPE}{per-user probability of error}
\newacronym{MAC}{MAC}{multiple access channel}
\newacronym{UMA}{UMA}{unsourced multiple-access}
\newacronym{TUMA}{TUMA}{type-based unsourced multiple-access}
\newacronym{TBMA}{TBMA}{type-based multiple access}
\newacronym{AWGN}{AWGN}{additive white Gaussian noise}
\newacronym{EP}{EP}{expectation propagation}
\newacronym{AMP}{AMP}{approximate message passing}
\newacronym{PMF}{PMF}{probability mass function}
\newacronym{PDF}{PDF}{probability density function}
\newacronym{CDF}{CDF}{cumulative distribution function}
\newacronym{SNR}{SNR}{signal-to-noise ratio}
\newacronym{GMAC}{GMAC}{Gaussian multiple access channel}
\newacronym{mMTCs}{mMTCs}{massive machine-type communications}
\newacronym{CCS-AMP}{CCS-AMP}{coded compressed sensing with approximate message passing}
\newcommand{\funceil}[1]{\left \lceil #1 \right \rceil}
\renewcommand{\define}{=}
\newcommand{\n}{{\uN}} %{n} %
\newcommand{\Ka}{{\uK_{\rm a}}} %{{K_{\rm a}}} %
\newcommand{\Kap}{{\uK'_{\rm a}}} %{{K_{\rm a}}} %
\newcommand{\Ma}{{\M_{\rm a}}} %{{M_{\rm a}}} %
\newcommand{\power}{\uP} %{P} %
\newcommand{\M}{\uM} %{M} %

\newcommand{\TV}[2]{\ensuremath{{\opT\opV}\big(#1,#2\big)}} % total variation distance 

\newcommand{\hT}{\hat{\bm{T}}}
\newcommand{\hN}{\hat{\bm{N}}}

\newcommand{\EbNoi}{E_{\rm b}/N_0}

\newtheorem{theorem}{Theorem}
\newcommand{\zv}{\mathbf{0}}

\newcommand{\eqD}[2]{#1 \stackrel{\rm D}{=}#2}
\newcommand{\indic}[1]{\mathbbm{1}\left\{#1\right\}}
\newcommand{\Su}[1]{\mathtt{Supp}\lefto(#1\right)}

\newcommand{\rveC}{\bm{C}}
\newcommand{\rvecN}{\bm{N}}
\newcommand{\nonneg}{\mathbb{N}_0}
\newcommand{\indep}
{\mathrel{\bot}\joinrel\mathrel{\mkern-5mu}%
                    \joinrel\mathrel{\bot}}
                    
\def\BibTeX{{\rm B\kern-.05em{\sc i\kern-.025em b}\kern-.08em
    T\kern-.1667em\lower.7ex\hbox{E}\kern-.125emX}}
\begin{document}
\title{An Achievability Bound for Type-Based Unsourced Multiple Access}

% %%% Single author, or several authors with same affiliation:
% \author{%
%  \IEEEauthorblockN{Author 1 and Author 2}
% \IEEEauthorblockA{Department of Statistics and Data Science\\
%                    University 1\\
%                   City 1\\
%                  Email: author1@university1.edu}% }

%%% Several authors with up to three affiliations:
\author{%
\IEEEauthorblockN{Deekshith Pathayappilly Krishnan\IEEEauthorrefmark{1}, Kaan Okumus\IEEEauthorrefmark{1}, Khac-Hoang Ngo\IEEEauthorrefmark{2}, Giuseppe Durisi\IEEEauthorrefmark{1}}
\IEEEauthorblockA{\IEEEauthorrefmark{1}Department of Electrical Engineering,
	Chalmers University of Technology,
	41296 Gothenburg, Sweden\\
	Email: \{deepat, okumus, durisi\}@chalmers.se}
\IEEEauthorblockA{\IEEEauthorrefmark{2}Department of Electrical Engineering  (ISY),
Link\"{o}ping University, {58183
Link\"{o}ping}, Sweden\\
Email: khac-hoang.ngo@liu.se} 
\thanks{This work was supported in part by the Swedish Research Council under
grants 2021-04970 and 2022-04471, and by the Swedish Foundation for
Strategic Research. The work of K.-H. Ngo was supported in part by the Excellence Center at Linköping--Lund in Information Technology (ELLIIT).}
}

\maketitle

%%%%%%
%% Abstract: 
%% If your paper is eligible for the student paper award, please add
%% the comment "THIS PAPER IS ELIGIBLE FOR THE STUDENT PAPER
%% AWARD." as a first line in the abstract. 
%% For the final version of the accepted paper, please do not forget
%% to remove this comment!
%%

\begin{abstract}
	We derive an achievability bound to quantify the performance of a {}{type-based unsourced multiple access} system---an information-theoretic model for grant-free multiple access with correlated messages. The bound extends available achievability results for the per-user error probability in the unsourced multiple access framework, where, different from our setup, message collisions are treated as errors. {}{Specifically,} we provide an upper bound on the total variation distance between the type ({i.e.,} the empirical {}{probability mass function}) of {}{the} transmitted messages and its estimate over a Gaussian multiple access channel. {}{Through numerical {simulations,} we illustrate that our bound can be used to determine the message type that is less efficient to transmit, because more difficult to detect. We finally show that a practical scheme for type estimation{,} based on coded compressed sensing with approximate message passing, operates approximately {}{3 dB} away from the {bound}, for the parameters considered in the paper.}
\end{abstract}

\section{Introduction}
 {}{Next generation wireless communication systems should support accurate decision-making from distributed network data~\cite{wang2023road}}. Typically, decision-making {}{is performed} at a central node, to which data-collecting entities (e.g., \gls{IoT} sensors) transmit their information. {}{The central node is often interested in a function of the received data}. For instance, applications like over-the-air aggregation in wireless federated learning \cite{gafni2024federated}, point cloud transmission~\cite{bian2024wireless}, and {}{majority-vote} computation~\cite{csahin2024over} involve estimating the type, i.e., {}{ the empirical \gls{PMF} of the received data}. {In this paper, we establish} a performance bound for communication systems employing type-based decision-making. {}{Specifically, we generalize the information-theoretic analysis {of} the \gls{UMA} model presented in}~\cite{PolyanskiyISIT2017massive_random_access} {}{to the case of \gls{TUMA}.}
\par {{In} \gls{UMA}, all transmitters share a common codebook {as well as} the communication medium (e.g., an additive white Gaussian noise channel), and the receiver produces a list of the transmitted codewords. In~\cite{PolyanskiyISIT2017massive_random_access}, an achievability bound on the {}{per-user error probability} is {}{obtained} under the assumption that {the} active transmitters select their message uniformly at random. In the analysis, the event that two {}{or more} users select the same codeword (message collision) is modeled as an error. Indeed, the probability that this occurs is typically negligible for the parameters considered in~\cite{PolyanskiyISIT2017massive_random_access} (hundreds of active users and $2^{100}$ messages). On the contrary, in this paper, we are interested in the case in which multiple users transmit the same message and the receiver is {}{tasked with estimating the set of transmitted messages along with their multiplicities}, {}{i.e., the number of users sending each message}.}
\par {Type-based {multiple access} was originally introduced in~\cite{Mergen2006} in the context of parameter estimation from distributed data. The \gls{TUMA} framework has been recently introduced in~\cite{hoang2023type}, in the context of multi-target tracking. There{}{,} the authors analyze the performance of different compressed-sensing-inspired algorithms for type estimation. However, no performance bound is derived. {}{This work addresses such a gap.}}
{\paragraph*{Contributions} We derive a numerically computable upper bound on the total variation distance between the transmitted message type and the estimated {}{type} for the \gls{TUMA} model. Our analysis relies on some of the tools used in~\cite{PolyanskiyISIT2017massive_random_access}, {}{such as} Chernoff's bound and Gallager's $\rho$-trick. The novel challenge in the \gls{TUMA} setting arises from the presence of additional error events, such as the {}{incorrect} estimation of message multiplicities, which make the parametrization and partitioning of error events more complex. In contrast, in the \gls{UMA} case{}{,} the error events can be straightforwardly parametrized by the number of misdetected or impostor (i.e., {}{false-positive}) messages, enabling {}{a simple} partitioning of {}{the} error events and {}{a} subsequent application of Gallager’s $\rho$-trick. In the \gls{TUMA} setting, the message multiplicity errors make the application of Gallager’s $\rho$-trick more nuanced.}
\par {}{{}We use our bound to numerically characterize, for a selected family of \gls{TUMA} message types, {}{and} for a fixed number of active users, the minimum energy per bit ({}{$\EbNoi$}) required to achieve a target estimation error, as a function of a distance metric between the considered type and an \gls{UMA} uniform type.} We measure this distance in terms of total variation.  Our analysis reveal{}{s} that the required $\EbNoi$ is maximized for intermediate distance values and exceeds the $\EbNoi$ required in the UMA case. However, as the distance approaches one, the $\EbNoi$ reduces significantly, and drops below the {}{value required} for \gls{UMA}, {}{due to the reduced number of multiplicities to detect and the higher power allocated per transmitted message.} Finally, we present an adaptation to \gls{TUMA} of the coded compressed sensing with approximate message passing (CCS-AMP) scheme proposed in \cite{Amalladinne2020}. The gap between the $\EbNoi$ value predicted by our bound and the {}{value} required by the proposed {CCS-AMP} {}{scheme} is around 3 dB.
\par {}{\paragraph*{Notation} We denote scalar system parameters by uppercase non-italic
letters, scalar random variables by uppercase italic letters, deterministic scalars by lowercase italic letters, vectors by uppercase boldface italic  letters, and deterministic vectors by lowercase boldface italic letters,  {}{e.g., $\Ka$, $X$, $x$, $\rvecx$, and $\vecx$, respectively}. We {}{let} $\matidentity_\n$ be the~$\n \times \n$ identity matrix and $\zv$ be the {}{all}-zero vector.  We denote the $\smash{\n\text{-dimensional}}$ Euclidean space by~$\reals^\n$, the set of naturals as $\mathbb{N}$, $\mathbb{N}\cup\{0\}$ as $\nonneg$, and its $\n$-fold product as $\nonneg^\n$.
We set $[m:n]=\left\{m,\hdots,n\right\}$, $m\le n$, and $m$, $n\in \mathbb{N}$; if $m=1$, $[m:n]=[n]$. For $x\in \reals$, $\funceil{x}$  is its ceiling value and $x^+=\max\{0,x\}$.  For $\vecx \in \reals^\n$ and $\setS\subseteq [\n]$, $\vecx_{\setS}$ denotes {}{the restriction of the vector to $\setS$; its dimension is $|\setS|$}. {}{For a vector $\vecn$}, {}{$\Su{\vecn}$ is} the set of its non-zero entries. {}{We denote the Gaussian distribution with mean $\zv$ and covariance matrix $\matA$ by~$\normal(\zv, \matA)$.}
We use~$\vecnorm{\cdot}$ for {}{the}~$\ell_2$-norm and $\vecnorm{\cdot}_1$ for {}{the}~$\ell_1$-norm.
We use $\sim$ to specify the distribution {}{of a random variable}, $\indep$ for independence, and~$\stackrel{D}{=}$~for equality in distribution.  We denote the indicator function as $\indic{\cdot}$, the total variation distance  as $\TV{\cdot}{\cdot}$, and use natural logarithms. 
\arxivonly{\par {\emph{Reproducible Research:} The codes used to generate the plots and results in this paper are available at \url{https://github.com/deekshithpk/TUMA-Bound.git}.}}
\section{System Model}
We consider a scenario wherein a large number of transmitters (users), of which only $\Ka$ are \textit{active},  communicate {}{with} a common receiver, over {}{$\n$ uses of a real-valued} {\gls{GMAC}}. We assume that $\Ka$ is  fixed and known to the receiver. The transmitters communicate over $\n$ channel uses. {}{Such devices could be, for instance, sensors deployed to track the states of certain targets}. The total number of states could  potentially be large and each state corresponds to a {}{message} chosen by a transmitter. We denote the total number of messages by $\M$. However, we assume that only $\Ma$ out of $\M$ messages are {}{active,} and the {}{active message set} is denoted as $\mathcal{M}_{\rm a}$. In {}{a} {}{target-tracking} scenario, the active messages could correspond to {}{target} locations within a small region {}{near} the sensors. We assume that each user chooses a message from $\mathcal{M}_{\rm a}$. This results in a type over $[\Ma]$, {}{to be defined shortly.} We assume  $\Ma \ll \M$ and $\Ma \le \Ka$, e.g., $\Ma=10$, $\Ka=10^2$, and $\M=10^4$. {{}{Furthermore, we assume {}{that} $\Ma$ {}{is} known but $\mathcal{M}_a$ {}{is} unknown to the receiver.}
\par {}{The channel output is given by}
\begin{IEEEeqnarray}{c}
	\rvecy=\sum_{ k=1}^{\Ka}\vecx_{ k}+\rvecz
	\label{Eq:Channel_Model}
\end{IEEEeqnarray}
where $\rvecz\sim \normal(\zv,\matI_\n)$ is the \n-dimensional noise vector,  independent of {}{the} channel input vector $\vecx_k$ of user $k$, {}{$\smash{k \in[\Ka]}$,} and $\rvecy$ is the channel output. We assume that the channel inputs satisfy a maximum power constraint{
		\begin{IEEEeqnarray}{c}
			\vecnorm{\vecx_{k}}^2 \le \n\power,~k\in[\Ka].
			\label{Eq:Max_pow_constraint}
		\end{IEEEeqnarray}}

\par In our setting, multiple users may {}{transmit} the same message to the receiver. For instance, in the multi-target tracking problem, multiple sensors {}{might} report the same location {}{of} a tracked target. We utilize an {}{$\Ma$-sparse} vector in {}{$\nonneg^\M$}  to describe the list of  {}{transmitted/decoded messages and their multiplicities. We refer to this vector as a multiplicity vector.} Specifically, let $\vecn=\tp{[n_1,\hdots,n_\M]}$ be the multiplicity vector corresponding to {}{the} messages transmitted by the users. If we denote by $W_k$  the message chosen by {}{user~$k$}, we have that} 
\begin{IEEEeqnarray}{rCl}
	n_m &=& \sum_{ k=1}^{\Ka}\mathbbm{1}\left\{W_{ k}= m\right\}.
\end{IEEEeqnarray}
We have that $\vecnorm{\vecn}_1=\Ka$; furthermore, its support 
\begin{IEEEeqnarray}{c}
    \textstyle{\Su{\vecn} = \left\{ m \in [\M]: n_m>0\right\}}
\end{IEEEeqnarray}
{has cardinality} $\Ma$.
For the \gls{UMA} case~\cite{PolyanskiyISIT2017massive_random_access}, $n_m=1$ for all transmitted messages,  and $\vecnorm{\vecn}_1=|\Su{\vecn}|$.
\par The decoder observes $\rvecy$ and obtains an estimate $\smash{\hN=\tp{[\hat{N}_1,\hdots,\hat{N}_\M]}}$ of $\vecn$. We assume {}{that} the decoder's estimate is subject to a constraint on the {}{average total variation distance}. Specifically, {}{given} the types $\smash{\vect=\vecn/\Ka}$ and $\smash{\hT=\hN/\Ka}$,  the decoder's {}{estimate} {}{must satisfy} $\Exop\big[\TV{\vect}{\hT}\big]\le \epsilon$, {}{where the expectation is taken over the distributions of the channel input and the noise.} Here,
\begin{IEEEeqnarray}{c}
	\TV{\vect}{\hT}=\frac{1}{2}\sum_{m=1}^{\M}\vert t_m-\hat{T}_m\vert
\end{IEEEeqnarray}
is the total variation distance, which coincides with the PUPE metric in \cite{PolyanskiyISIT2017massive_random_access} when $\Ma = \Ka$; a more general metric is considered in \cite{hoang2023type}. {We shall refer to the setting just introduced as {}{the} \gls{TUMA} setting. Now, we provide a formal definition of a code for this setting. {}{To this end},  it {}{is} convenient to define the set of multiplicity vectors on $[\M]$ as}  \[\mathfrak{P}=\left\{\vecn \in \nonneg^\M:|\Su{\vecn}|=\Ma,~\vecnorm{\vecn}_1=\Ka\right\}.\]
\begin{definition}[\textit{TUMA Code}]
 An $( \M,\n,\power,\epsilon)$ \gls{TUMA} {}{code} for a {}{GMAC} with $\Ka$ active users and $\Ma$ active messages resulting in the message multiplicity vector $\vecn$, consists of an encoder $\xi:[\M] \mapsto \mathbb{R}^\n$ that produces a channel input satisfying \eqref{Eq:Max_pow_constraint}  and a decoder  $\psi:\reals^\n \mapsto \mathfrak{P}$ that {}{produces a multiplicity vector based on} the output $\rvecy$ of the \gls{GMAC}, {}{and satisfies} the constraint $\Exop[\TV{\vecn/\Ka}{\psi(\rvecy)/\Ka}] \le \epsilon$.
\end{definition}
\section{A {}{Random-Coding} Bound for TUMA}
\par {}{Our main result is a random-coding achievability bound stated in the next theorem.}
\begin{theorem}[\emph{Random-Coding Bound}] For {}{a} \gls{GMAC} with $\Ka$ active users and $\Ma$ active messages, and a fixed message multiplicity vector  $\vecn$ with $|\Su{\vecn}|=\Ma$ and $\vecnorm{\vecn}_1=\Ka$, there exists a $\lefto(\M,\n,\power,\epsilon\right)$ TUMA code {}{for which}
	\begin{IEEEeqnarray}{rCl}
		\epsilon \le \sum_{ t=1}^{\Ka}\frac{ t}{\Ka}\tilde{p}_{ t}+p_0+p_1.
		\label{Eq:Final_bound}
	\end{IEEEeqnarray}
    Here,
	\begin{IEEEeqnarray}{rCl}
		p_0&=&\Ma \Prob\lefto[ \vecnorm{  \rveC_{1}}^2 >\n\power\right] \label{Eq:Def_p0}\label{Eq:Def_p_0}
        \end{IEEEeqnarray}
        with $\rveC_1\sim \normal\lefto(\zv,\power'\matI_\n\right)$, and $0<\power'<\power$,
        \begin{IEEEeqnarray}{rCl}
          p_1&=& e^{-\frac{\n\delta^2}{8}}\label{Eq:Def_p_1}  
        \end{IEEEeqnarray}
	    with
        \begin{IEEEeqnarray}{rCl}
        \delta &\in& \lefto(0,\frac{1}{1-2\delta_{\min}^*}-1\right)
		\label{Eq:delta_interval}
        \end{IEEEeqnarray}
        and
        \begin{IEEEeqnarray}{c}
            \delta_{\min}^*= \frac{1}{2}-\min_{t,\setS,\ell,i,j,\rho}\lambda^*
		\label{Eq:Def_delta_start_min}
        \end{IEEEeqnarray}
        where $\lambda^*\in \lefto(0,{1}/{2}\right)$ is defined in \eqref{Eq:lambda_star}. Finally,
        \begin{IEEEeqnarray}{rCl}
		\tilde{p}_t &=& \sum_{\ell=0}^{\Ma}\sum_{\setS \in \mathfrak{S}_\ell}\sum_{i=0}^{\Ma-\ell}\sum_{\setN \in \mathfrak{N}_i}\sum_{j=\ell}^{t-\Ma+\ell+i}e^{\rho \n R-\n E_0}.\label{Eq:Def_EN} \label{Eq:Def_tilde_p}
        \end{IEEEeqnarray}
        In both \eqref{Eq:Def_delta_start_min} and \eqref{Eq:Def_EN}, $\smash{\rho\in [\rho_{\min},1]}$ for a given $\smash{\rho_{\min}\in(0,1)}$, and $t$,~$\setS$,~$\ell$,~$i$,~$j$~in~\eqref{Eq:Def_delta_start_min} are optimized over the range of values given in~\eqref{Eq:Final_bound} and~\eqref{Eq:Def_tilde_p},
        \begin{IEEEeqnarray}{rCl}
          \mathfrak{S}_\ell&=&\{\setS\subseteq [\Ma]:|\setS|=\ell,\vecnorm{\vecn_{\setS}}_1 \le t\}\label{Eq:Def_MisSet}\\
          \mathfrak{N}_i&=&\{\setN\subseteq [\Ma]\setminus\setS:|\setN|=i, \vecnorm{\vecn_{\setN}}_1 \ge t+i-\vecnorm{\vecn_{\setS}}_1\}.\nonumber\\
          &&\label{Eq:Def_DetDeflSet}
        \end{IEEEeqnarray}
        Furthermore, 
        \begin{IEEEeqnarray}{rCl}
		R&=&\frac{1}{\n}\log \M'+\frac{1}{\n}\log \frac{\M^\ell}{\ell!} \label{Eq:Def_R}\\
		E_0&=&\frac{1}{2}\log\left(1-2b\rho\right)+ \rho a \label{Eq:E0}\label{Eq:Def_E0}\\
		\label{Eq:M_prime} \M'&=&{j-1 \choose \ell-1}{(t-j-1)^+ \choose {(\Ma-\ell-i-1)^+ }}\min\lefto\{\M_0^+,\M^+\right\}\\
		\label{Eq:M_0_plus}\M_0^+&=&\binom{t-\vecnorm{\vecn_{{\setS}}}_1+(i-1)^+}{{(i-1)^+}}\\
\label{Eq:M_plus}\M^+&=&\binom{(\vecnorm{\vecn_{{\setS}}}_1+\vecnorm{\vecn_{{\setN}}}_1-t-1)^+}{(i-1)^+} \\
		a&=&\frac{1}{2}\log\lefto(1+2\power'{c}_{\min} \lambda^*\right) \label{Eq:Def_a}\\
		b&=&\lambda^* \lefto( 1-\frac{1}{1+2{c_{\min}}\power'\lambda^*}\right)\label{Eq:Def_b}\\
		\lambda ^*&=&\frac{(c_{\min}\power'-2)+\sqrt{(c_{\min}\power'-2)^2+4c_{\min}\power'\lefto(1+\rho\right)}}{4c_{\min}\power'\lefto(1+\rho\right)}\nonumber\\
		&& \label{Eq:lambda_star} \\
c_{\min}&=&\mathbbm{1}\{\ell>0\}\funceil{\frac{\vecnorm{\vecn_{\setS}}_1^2}{\ell}}+\mathbbm{1}\{i>0\}\funceil{\frac{\lefto(t-\vecnorm{\vecn_{\setS}}_1\right)^2}{i}}\nonumber\\
    &&+\mathbbm{1}\{t>j\}\funceil{\frac{\lefto(t-j\right)^2}{\Ma-\ell-i}}+\mathbbm{1}\{\ell>0\}\funceil{\frac{j^2}{\ell}}.\label{Eq:c_min_defn}
	\end{IEEEeqnarray}
	\label{Th:Main}
\end{theorem}
%It reduces the impact of events with probabilities much greater than one at the expense of those near zero, both being atypical and less significant.
\par {}{The proof of the theorem has a structure similar to the proof of\mbox{\cite[Theorem 1]{PolyanskiyISIT2017massive_random_access}}, in which an \gls{UMA} scenario is considered. Specifically, it relies {}{on} random coding, Chernoff{}{'s} bound, and the use of Gallager's $\rho$-trick\mbox{\cite[Eq. (2.28)]{Gallager1968information}}. The use of Gallager's $\rho$-trick is more delicate in the \gls{TUMA} case. {}{Indeed, in the UMA setting}, the error events, i.e., the {misdetection} of some messages and the inclusion of impostor messages{}{,} are such that the number of misdetected messages coincides with {}{the number of} impostor messages. This enables a convenient partitioning of the error event space{}{,} {}{achieved} by parametrizing the error events through the number of misdetected (or impostor) messages. This, {}{in turn}, facilitates the application of Gallager's $\rho$-trick.
\par {}{In \gls{TUMA}, however, the space of error events {}{is} much {}{larger}. Indeed{}{,} a \gls{TUMA} decoder can incorrectly estimate the multiplicities of correctly decoded messages, either inflating or deflating them. As a consequence, not even the sum of multiplicities of the misdetected messages coincide{}{s} with that of the impostor messages. This necessitates a more nuanced parametrization of {}{the} error events, {}which needs to account for both the number and the position of the misdetected messages, impostor messages, and inflated and deflated detected messages to effectively partition the error event space. {Using this novel parametrization, we upper-bound the probability of error events in each partition by analyzing {}{its} {}{Chernoff's} exponent, conditioned on a high-probability set. This leads to the term $p_1$ in~\eqref{Eq:Def_p_1}. Furthermore, to simplify the analysis, we apply Gallager's $\rho$-trick only once, {unlike in}~\cite{PolyanskiyISIT2017massive_random_access}, where it is {}{applied} twice.}
\par {}{One final challenge is that the chosen performance metric, i.e., the total variation distance, is a nonlinear function of the multiplicity vectors. This nonlinearity complicates the estimation of the total number of error events contributing to the union bound summation in Gallager's $\rho$-trick. While in {}{the} UMA {}{case} this number can be estimated {}{using} a straightforward counting argument, our estimation involves linearizing the total variation distance and leveraging expressions for integer-valued solutions to partition equations {}{(see Lemma} \ref{Lem:Pi_Cardinality}).}
\par\confonly{Now, we present a sketch of the proof of Theorem \ref{Th:Main}. The full proof can be found in \cite{TUMA2025}.}\arxivonly{Now, we present the proof of Theorem \ref{Th:Main}.}
\begin{proof}[\confonly{Proof Sketch of Theorem \ref{Th:Main}}\arxivonly{Proof of Theorem \ref{Th:Main}}]
	\arxivonly{We present the proof in multiple steps.} First, we state some supporting lemmas. \confonly{The proof of these lemmas can be found in the extended version~\cite{TUMA2025}.}\arxivonly{The proof of these lemmas are delegated to the appendices.}
	\arxivonly{\par In Lemma \ref{Lem:Delta} below, we prove that the $\ell_1$~distance between the transmitted multiplicity vector and the decoded multiplicity vector is even and does not exceed~{$2\Ka$}.
    \begin{lemma}
			\label{Lem:Delta}
			Fix any two multiplicity vectors $\vecn$ and $\hat{\vecn}$ such that $\vecnorm{\vecn}_1=\vecnorm{\hat{\vecn}}_1=\Ka$, $\vecn \neq \hat{\vecn}$, $\smash{\Su{\vecn}\subset[\M]}$, $\smash{\Su{\hat{\vecn}}\subset[\M]}$, $|\Su{\vecn}|=|\Su{\hat{\vecn}}|=\Ma$ and $\smash{\Ka \ge \Ma}$. Then, $\vecnorm{\vecn-\hat{\vecn}}_1=2t$, for some $t\in[\Ka]$.
		\end{lemma}
		\begin{proof}
			See Appendix \ref{Append:Scaled_TV_Even}.
		\end{proof}
		}
	\par {}{In Lemma \ref{Lem:Opt_lambda} below, we solve an optimization problem that will turn out {}{to be} useful in analyzing the {Chernoff's} exponent.}
	\begin{lemma}
		Let $\rho\in(0,1]$, $\power'>0$, and {$c>0$}. Define
		\begin{IEEEeqnarray}{rCl}
			E_0'\lefto(\lambda\right)&=&\frac{\rho}{2}\log\big(1+2\power'c\lambda\big)\nonumber\\
			&&+\frac{1}{2}\log\lefto(1-2\rho\lambda\lefto(1-\frac{1}{1+2\power'c\lambda}\right)\right)
			\label{Eq:Def_E0_prime}
		\end{IEEEeqnarray}
		{}{ where $\lambda>0$ } satisfies $1-2\rho\lambda\lefto(1-{1}/{1+2\power'c\lambda}\right)>0.$ Then, $E_0'$ is maximized at $\lambda^*$, and $\lambda^* <1/2$.
		\label{Lem:Opt_lambda}
	\end{lemma}
	\arxivonly{\begin{proof}
			See Appendix \ref{Sec:Opt_Chernoff}.
		\end{proof}}
	{}{In Lemma \ref{Lem:Exp_f_c} below{}{,} we prove that a certain function is monotonically decreasing. We will use this result to upper{}{-}bound an exponential term arising from Chernoff's bound.}
	\begin{lemma}
		For $\n \in \mathbb{Z}_+$ and $\delta\in (0,1)$, let $\vecz\in\reals^{\n}$ be such that $\vecnorm{\vecz}^2 \le \n(1+\delta)$. {}{Furthermore, let} $\rho\in(0,1]$, $\power'>0$, and $\lambda^*$ as in \eqref{Eq:lambda_star}. Define
		\begin{IEEEeqnarray}{rCl}
			f(c)&=&\lambda^*\lefto(\vecnorm{\vecz}^2-\frac{\vecnorm{\vecz}^2}{1+2c \power'\lambda^*}\right)-\frac{\n}{2}\log \lefto(1+2c\power'\lambda^*\right).\nonumber\\
			&&\label{Eq:df_dc}
		\end{IEEEeqnarray}
		Then, for $c>0$, $f(c)$ is monotonically decreasing in $c$.
		\label{Lem:Exp_f_c}
	\end{lemma}\arxivonly{
		\begin{proof}
			See Appendix \ref{Append:Exp_f_c}.
		\end{proof}}
	Finally, we present a lemma {}{in which we bound} the number of multiplicity vectors {}{at} a given {}{$\ell_1$} distance {}{from} the transmitted multiplicity vector. {}{Furthermore, we provide} a lower bound on the {}{$\ell_2$} distance between any {}{decoded multiplicity vector and the transmitted multiplicity vector}.
	\begin{lemma}
		Fix a multiplicity vector $\vecn$ such that $\vecnorm{\vecn}_1=\Ka$, $|\Su{\vecn}|=\Ma$,  $\Su{\vecn}\subset[\M]$, and $\Ka \ge \Ma$. Fix $\smash{t\in [\Ka]}$. Let  $\hat{\vecn}$  denote a generic multiplicity vector such that $\vecnorm{\hat{\vecn}}_1=\Ka$, $|\Su{\hat{\vecn}}|=\Ma$, $\Su{\hat{\vecn}}\subset[\M]$.   {}{Furthermore}, let $\smash{\setS=\Su{\vecn}\cap \Su{\hat{\vecn}}^c}$, $\smash{|\setS|=\ell}$,~$\smash{\vecnorm{\vecn_{\setS}}_1\ge t}$,~$\smash{\widehat{\setS}=\Su{\vecn}^c\cap \Su{\hat{\vecn}}}$,~$\smash{|\hat{\setS}|=\ell}$, $\smash{\vecnorm{\hat{\vecn}_{\hat{\setS}}}_1=j}$,
		\begin{IEEEeqnarray}{c}
			\setN=\{m\in \Su{\vecn}\cap \Su{\hat{\vecn}}: n_m \ge \hat{n}_m\}
		\end{IEEEeqnarray}
		with $|\setN|=i$, and $\widehat{\setN}=\Su{\vecn}\cap \Su{\hat{\vecn}}\cap \setN^c$. Then, {}{for fixed $\ell$, $\setS$, $i$, $\setN$, and $j$,} the total number of {}{vectors} $\hat{\vecn}$ {}{satisfying} $\vecnorm{\vecn-\hat{\vecn}}_1=2t$ is at most $\M'$, where $\M'$ is defined in \eqref{Eq:M_prime}.
		{}{Furthermore}, such an $\hat{\vecn}$ satisfies $ \vecnorm{\vecn-\hat{\vecn}}^2 \ge c_{\min}$,
		with $c_{\min}$ defined in~\eqref{Eq:c_min_defn}.
		\label{Lem:Pi_Cardinality}
	\end{lemma}\arxivonly{
		\begin{proof}
			See Appendix \ref{Append:Cardinality_Pi}.
		\end{proof}}
	{The proof of the above lemma involves linearizing the equation $\vecnorm{\vecn-\hat{\vecn}}_1=2t$, and counting the integer valued solutions of the resulting partition equations to obtain $\M'$.}
    \subsubsection{Codebook Generation}
	  For $\smash{m\in [\M]}$ and $\smash{\power'<\power}$, {}{we} generate {}{the} codewords $\smash{\rveC_{\rm m}\sim \normal(\zv,\power'\matI_\n)}$ {}{independently}.  To send message $W_k$, {}{the} transmitter chooses {}{codeword} $\rveC_{W_k}$ and {}{set} $\smash{\rvecx_{k}=\indic{\|\rveC_{W_{ k}}\|^2 \le \n\power}\rveC_{W_{k}}}$.
    \subsubsection{Decoder} The decoder outputs {}{the multiplicity vector} $\hat{\bm{N}}\in \mathfrak{P}$  minimizing $\smash{\vecnorm{\rvecy-\rveC(\hN)}^2}$, where
    \begin{IEEEeqnarray}{c}
        \rveC(\hat{\bm{N}})= \sum_{ m\in \Su{\hat{\bm{N}}}}\hat{N}_{ m} \bm{C}_{ m}.
    \end{IEEEeqnarray}
 Note that the probability {}{that a tie occurs} is zero. 
 \par Next, we analyze the {}{average total variation between $\smash{\vect=\vecn/\Ka}$ and $\hat{\rvect}=\hat{\rvecN}/\Ka$ where the average is with respect to the noise vector $\rvecz$ in \eqref{Eq:Channel_Model}, and the random codebook.} 
    \subsubsection{Error Analysis} Similar to \cite{PolyanskiyISIT2017massive_random_access}, we {}{replace} the joint probability distribution over which $\Exop[\TV{\vect}{\hat{\bm{T}}}]$ is {}{computed with} the distribution for which $\rvecx_k\stackrel{D}{=}\rveC_{W_k}$. This can be done at the expense of adding an additional total variation {}{term} {}{(the $p_0$ term in} \eqref{Eq:Def_p_0}). {}{Next}, we  evaluate $\Exop[\TV{\vect}{\hat{\bm{T}}}]$ conditioned on the \textit{high probability event} $\smash{\setA=\big\{\vecnorm{\rvecz}^2  \le \n+\n\delta \big\}}$, {}{where $\delta$ is chosen as in}~\eqref{Eq:delta_interval}. {}{This yields the additional penalty $p_1$ defined in \eqref{Eq:Def_p_1}, which is obtained by using  the concentration bound for the sum of chi-squared random variables given in\mbox{\cite[Eq.~(2.19)]{wainwright2019high}}.}
	\arxivonly{From Lemma \ref{Lem:Delta}, we can express $\Exop\big[\TV{\vect}{\hT}\indic{\rvecz\in \setA}\big]$ as}\confonly{\par It can be seen that}
	\begin{IEEEeqnarray}{rCl}
		\Exop\big[\TV{\vect}{\hT}\indic{\rvecz\in \setA}\big]=\sum_{t=1}^{\Ka} \frac{t}{\Ka}p_t
		\label{Eq:TV_as_Exp_Delta_Eq_2t}
	\end{IEEEeqnarray}
	where $p_{t}=\Prob[ \vecnorm{\vecn-\hN}_1=2  t,\rvecz\in\setA]$. {}{Next}, we will upper-bound {}{$p_t$}.
	\par  Without loss of generality, we assume $\setM_{\rm a}=[\Ma]$. We define the set of all decoded message multiplicity vectors at an {}{$\ell_1$} distance of $2 t$ from $\vecn$ as
	\begin{IEEEeqnarray}{rCl}
		\hat{\Pi}_{ t}(\vecn)&=&\lefto\{\hat{\vecn}: \vecnorm{\vecn-\hat{\vecn}}_1=2 t\right\}.
	\end{IEEEeqnarray}
	{}{Set} $\rveC(\vecn)= \sum_{ m\in \Su{\vecn}}n_m \mathbf{C}_m$ , $c(\hat{\vecn}) = \vecnorm{\vecn-\hat{\vecn}}^2$, {}{and
			\begin{IEEEeqnarray}{c}
				\setE(\hat{\vecn})=\lefto\{\vecnorm{\rvecz+\sqrt{\power'c(\hat{\vecn})}\rvecx '}< \vecnorm{\mathbf{Z}}\right\}.
				\label{Eq:Def_E_hatn}
			\end{IEEEeqnarray}
		}
        We observe that
	\arxivonly{
		\begin{IEEEeqnarray}{rCl}
			p_{ t}&=&\mathbb{P}\lefto[\bigcup_{\hat{\vecn}\in \hat{\Pi}_{ t}(\vecn)}\lefto\{\text{Decoded vector is}~{\hat{\vecn}}\right\}\cap \lefto\{\rvecz \in \setA\right\}\right]\\
			&=&\Prob\lefto[\bigcup_{\hat{\vecn}\in \hat{\Pi}_{ t}(\vecn)} \lefto\{\vecnorm{\rvecy-\rveC(\hat{\vecn})}< \vecnorm{\rvecy-\rveC\lefto(\vecn\right)}\right\}\cap \lefto\{\rvecz \in \setA\right\}\right]\nonumber\\
			&&\\
\label{Eq:Y_eq_C_pl_Z}&=&\Prob\lefto[\bigcup_{\hat{\vecn}\in \hat{\Pi}_t(\vecn)} \lefto\{\vecnorm{\rvecz+\rveC(\vecn)-\rveC(\hat{\vecn})}< \vecnorm{\rvecz}\right\}\cap \lefto\{\rvecz \in \setA\right\}\right]\\
			&=&\Prob\lefto[\bigcup_{\hat{\vecn}\in \hat{\Pi}_t(\vecn)} \lefto\{\vecnorm{\rvecz+\sqrt{\power'\vecnorm{\vecn-\hat{\vecn}}^2}\rvecx'}< \vecnorm{\rvecz}\right\}\cap \lefto\{\rvecz \in \setA\right\}\right]\nonumber\\
			\label{Eq:C_ind_Z}&&\\
			&=&\Prob\lefto[\bigcup_{\hat{\vecn}\in \hat{\Pi}_t(\vecn)} \lefto\{\vecnorm{\rvecz+\sqrt{\power'c(\hat{\vecn})}\rvecx'}< \vecnorm{\rvecz}\right\}\cap \lefto\{\rvecz \in \setA\right\}\right]
			\label{Eq:p_t_Defn}\\
            &=&\Prob\lefto[\bigcup_{\hat{\vecn}\in \hat{\Pi}_t(\vecn)}\setE(\hat{\vecn}) \cap \lefto\{\rvecz \in \setA\right\}\right].\label{Eq:p_t_via_En}
		\end{IEEEeqnarray}
		where $\eqref{Eq:Y_eq_C_pl_Z}$  follows from the fact that $\rvecy=\rveC(\vecn)+\rvecz$, and~$\eqref{Eq:C_ind_Z}$ follows by noting that $\rveC(\vecn)-\rveC(\hat{\vecn}) \indep \rvecz$,
		\begin{IEEEeqnarray}{c}
			\eqD{\rveC(\vecn)-\rveC(\hat{\vecn})}{\sqrt{\power'\vecnorm{\vecn-\hat{\vecn}}^2}}\rvecx'
			\label{Eq:Distr_Eq_1}
		\end{IEEEeqnarray}
		where $\rvecx' \sim \mathcal{N}(\zv,\matidentity_\n)$, and $\rvecx' \indep \rvecz$.}\confonly{
		\begin{IEEEeqnarray}{rCl}
			p_{ t}&=&\mathbb{P}\lefto[\bigcup_{\hat{\vecn}\in \hat{\Pi}_{ t}(\vecn)}\lefto\{\text{Decoded  vector is}~{\hat{\vecn}}\right\}\cap \lefto\{\rvecz \in \setA\right\}\right] \nonumber\\
			\arxivonly{&&\\
			&=&\Prob\Bigg[\bigcup_{\hat{\vecn}\in \hat{\Pi}_{ t}(\vecn)} \Big\{\vecnorm{\rvecy-\rveC(\hat{\vecn})}< \vecnorm{\rvecy-\rveC\big(\vecn\big)}\Big\}\cap \setA\Bigg]\nonumber\\
			&&\\
			&=&\Prob\Bigg[\bigcup_{\hat{\vecn}\in \hat{\Pi}_t(\vecn)} \Big\{\vecnorm{\rvecz+\sqrt{\power'c(\hat{\vecn})}\rvecx'}< \vecnorm{\rvecz}\Big\}\cap \setA\Bigg]\label{Eq:p_t_Exprssn}\\}
			&=&\Prob\Bigg[\bigcup_{\hat{\vecn}\in \hat{\Pi}_t(\vecn)}\setE(\hat{\vecn}) \cap \left\{\rvecz \in \setA\right\}\Bigg].\label{Eq:p_t_via_En}
		\end{IEEEeqnarray}
		Here, {}{to obtain} \eqref{Eq:p_t_via_En}, {}{we used that} $\smash{\rveC(\vecn)-\rveC(\hat{\vecn}) \indep \rvecz}$, and $\eqD{\rveC(\vecn)-\rveC(\hat{\vecn})}{\sqrt{\power'\vecnorm{\vecn-\hat{\vecn}}^2}}\rvecx'$, where $\rvecx' \sim \mathcal{N}(\zv,\matidentity_\n)$, and $\rvecx' \indep \rvecz$.}  
	Now, fixing  $\vecz \in \setA$, we analyze the conditional event {}{$\smash{\Prob[\setE(\hat{\vecn}) |\rvecz=\vecz]}$}. By invoking {}{Chernoff's bound} 
	\begin{IEEEeqnarray}{c}
		\Prob\Big[\setE(\hat{\vecn})\Big\vert {}{\rvecz=\vecz} \Big] \le \Exop \left[\exp\left(E_1\left(\rvecx ',c(\hat{\vecn}),\vecz\right)\right)\right]
	\end{IEEEeqnarray}
	where $\textstyle{E_1\left(\rvecx ',c(\hat{\vecn}),\vecz\right)=\lambda\vecnorm{\vecz}^2-\lambda\vecnorm{\vecz+\sqrt{\power'c(\hat{\vecn})}\rvecx '}^2.}$
	% \begin{IEEEeqnarray}{c}
	% 	E_1\left(\rvecx ',c(\hat{\vecn}),\vecz\right)=\lambda\vecnorm{\vecz}^2-\lambda\vecnorm{\vecz+\sqrt{\power'c(\hat{\vecn})}\rvecx '}^2.
	% \end{IEEEeqnarray}
	We choose $\lambda=\lambda^*$, where $\lambda^*$ is defined in \eqref{Eq:lambda_star}.  \arxivonly{(We make this choice of $\lambda$ so that the function {}{$E_0'$} defined in \eqref{Eq:Def_E0_prime}, and which comes up in our final expression in \eqref{Eq:Final_bound}, is maximized.)}
	\par Next, utilizing the identity
	\begin{IEEEeqnarray}{c}
		\label{Eq:Identity}
		\Exop\lefto[e^{-\theta \vecnorm{\sqrt{\alpha}\rvecz+\vecv}^2}\right]=\frac{\exp\lefto(-\frac{\theta\vecnorm{\vecv}^2}{1+2\alpha\theta}\right)}{(1+2\alpha\theta)^{\frac{\n}{2}}}
	\end{IEEEeqnarray}
	{}{which holds}  for $\alpha>0$, $\mathbf{v}\in\mathbb{R}^\n$, and $2\alpha\theta>-1$, we obtain

	\begin{IEEEeqnarray}{c}
		\Prob\Big[\setE(\hat{\vecn})\Big\vert {}{\rvecz=\vecz} \Big] \le \exp(f(c(\hat{\vecn})))
	\end{IEEEeqnarray}
	where $f(\cdot)$ is defined  {}{in} \eqref{Eq:df_dc}. {Note that the condition $\smash{2\alpha\theta>-1}$ is ensured in {}{our} setup, since $\lambda^*>0$, and $\smash{c(\hat{\vecn})>0}$.}
	% \par Denoting $\alpha_1=\lambda \vecnorm{\vecz}^2$, $\alpha_2=2 \power'\lambda$, $\alpha_3={\n}/{2}$, we have
	% \begin{IEEEeqnarray}{c}
	% f'(c)=\frac{\alpha_1\alpha_2}{(1+\alpha_2c)^2}-\frac{\alpha_3\alpha_2}{1+\alpha_2c}.  
	% \end{IEEEeqnarray}
	Now, invoking Lemma \ref{Lem:Pi_Cardinality}, we {}{conclude that}
	$\smash{c(\hat{\vecn}) \ge c_{\min}}$. {Furthermore,} since $\vecz \in \mathcal{A}$, {it follows} from Lemma \ref{Lem:Exp_f_c} {}{that} $\smash{f(c(\hat{\vecn}))\le f(c_{\min}).}$
	That is, for $\vecz \in \setA$, {}{we have}
	\begin{IEEEeqnarray}{rCl}
		\Prob\Big[\setE(\hat{\vecn})\Big\vert {}{\rvecz=\vecz} \Big]
		&\le& e^{\lambda^* f(c_{\min})}.
	\end{IEEEeqnarray}
	\par Now, we proceed {}{by} apply{}{ing}  Gallager's {}{$\rho$-trick}. We fix  $\ell$, $\setS$, $i$, $\setN$, {}{$\hat{\setS}$, $\hat{\setN}$, and~$j$} as in the statement of Lemma \ref{Lem:Pi_Cardinality}. For convenience, we set $\smash{\eta=(t, \ell, \setS, i, \setN, j)}$,  and denote by $\sum_{\eta}$ and {}{$\bigcup_{\eta}$} the multiple summations and unions over the indices, respectively. Let $\smash{c_1(\hat{\vecn}) =\vecnorm{\vecn_{\setS}}_1 + \vecnorm{\vecn_{\setN}}_1 - \vecnorm{\hat{\vecn}_{\setN}}_1 }$, $\smash{c_2(\hat{\vecn}) =  \vecnorm{\hat{\vecn}_{\hat{\setN}}}_1 - \vecnorm{\vecn_{\hat{\setN}}}_1 + j}$, and 
	\begin{IEEEeqnarray}{c}
	    \hat{\Pi}'(\vecn) = \Big\{\hat{\vecn} :
		c_1(\hat{\vecn})+ c_2(\hat{\vecn}) = 2t \Big\}
		\label{Eq:Pi_prime}
	\end{IEEEeqnarray}
			{}{and note that} $\smash{\bigcup_{\hat{\vecn}\in\hat{\Pi}_t(\vecn)}\setE(\hat{\vecn})=\bigcup_{\eta}\bigcup_{\hat{\vecn}\in\hat{\Pi}'(\vecn)}\setE(\hat{\vecn}).}$
	% \begin{IEEEeqnarray}{rCl}
	% 	\bigcup_{\hat{\vecn}\in\hat{\Pi}_t(\vecn)}\setE(\hat{\vecn})&=&\bigcup_{\eta}\bigcup_{\hat{\vecn}\in\hat{\Pi}'(\vecn)}\setE(\hat{\vecn}).\nonumber
	% 	\label{Eq:Def_Pi_hat_prime}
	% \end{IEEEeqnarray}
	{\arxivonly{Fix $\smash{p(\vecz,\hat{\vecn},\setA)}=\indic{\vecz \in \setA } \Prob\lefto[\setE(\hat{\vecn}) \vert \vecz\right]$.} Then, Gallager's $\rho$-trick applied to \eqref{Eq:p_t_via_En} yield{}{s}}
		\begin{IEEEeqnarray}{rCl}
			\Prob\lefto[\bigcup_{\hat{\vecn}\in \Pi_t(\vecn)}\setE(\hat{\vecn})\cap \setA \Big\vert \vecz \right]
&\le&\sum_{\eta}\Prob\lefto[\bigcup_{\hat{\vecn}\in \hat{\Pi}'(\vecn)}\setE(\hat{\vecn})\cap \setA \Big\vert \vecz \right]\\
			&\le & \sum_{\eta}\lefto(\sum_{\hat{\vecn}\in \hat{\Pi}'(\vecn)}\Prob\lefto[\setE(\hat{\vecn})\cap \setA \Big\vert \vecz \right]\right)^\rho\nonumber\\
			&&\\
            \arxivonly{&=& \sum_{\eta}\lefto(\sum_{\hat{\vecn}\in \hat{\Pi}'(\vecn)}p(\vecz,\hat{\vecn},\setA)\right)^\rho\\
		&\le & \sum_{\eta} \lefto(\sum_{\hat{\vecn}\in \hat{\Pi}'(\vecn)} e^{f(c_{\min})} \right)^\rho \\}
			&\le & \sum_{\eta}|\hat{\Pi}'(\vecn)|^\rho e^{\rho f(c_{\min})}.\label{Eq:c_min_uniformity}
		\end{IEEEeqnarray}
 %\arxivonly{Then, with the re-indexed union of the event $\setE(\hat{\vecn})$, we apply Gallager's $\rho-$trick to upper-bound {}{$p_t'$} in the following way.
% 	\begin{IEEEeqnarray}{rCl}
% p_t'&\le&\sum_{\eta}\Prob\Big[\bigcup_{\hat{\vecn}\in \hat{\Pi}'(\vecn)}\setE(\hat{\vecn})\cap \setA \Big\vert \vecz \Big]\\
% 		&\le & \sum_{\eta}\lefto(\sum_{\hat{\vecn}\in \hat{\Pi}'(\vecn)}\Prob\Big[\setE(\hat{\vecn})\cap \setA \Big\vert \vecz \Big]\right)^\rho\\
% 		&=& \sum_{\eta}\lefto(\sum_{\hat{\vecn}\in \hat{\Pi}'(\vecn)}\indic{\vecz \in \setA } \Prob\lefto[\setE(\hat{\vecn}) \Big\vert \vecz\right] \right)^\rho\\
% 		&\le & \sum_{\eta} \lefto(\sum_{\hat{\vecn}\in \hat{\Pi}'(\vecn)}\indic{\vecz \in \setA } e^{f(c_{\min})} \right)^\rho \\
% 		&\le & \sum_{\eta}|\hat{\Pi}'(\vecn)|^\rho e^{\rho f(c_{\min})}.\label{Eq:c_min_uniformity}
% 		%&&+p_1.
% 	\end{IEEEeqnarray}}
	Here, \eqref{Eq:c_min_uniformity} follows because\confonly{~$\smash{\Prob[\setE(\hat{\vecn})\cap \setA \vert \vecz ]}$ equals $\smash{\indic{\vecz \in \setA  \Prob[\setE(\hat{\vecn}) \vert \vecz]}}$, and because} upon fixing~$\eta$, the parameter $c_{\min}$ is the same for all $\hat{\vecn}\in \hat{\Pi}'(\vecn)$. This, in turn, follows from the definition of $c_{\min}$ in \eqref{Eq:c_min_defn}, and the observation that, since $\vecnorm{\vecn}_1=\vecnorm{\hat{\vecn}}_1$, we have $\smash{c_1(\hat{\vecn})=c_2(\hat{\vecn})}=t$. {}{Using the definition of $f(\cdot)$} in~\eqref{Eq:df_dc}, we have
	\begin{IEEEeqnarray}{rCl}
		\Prob\Big[\bigcup_{\hat{\vecn} \in \hat{\Pi}'(\vecn)}\setE(\hat{\vecn})\cap \setA \Big\vert \vecz\Big] &\le&\lefto(   |\hat{\Pi}'(\vecn)| e^{ b \vecnorm{\vecz}^2-\n a}\right)^\rho
	\end{IEEEeqnarray}
	where $a$ and $b$ are defined as in \eqref{Eq:Def_a} and \eqref{Eq:Def_b}, respectively.
	\par Now, we take the expectation over $\rvecz$. We observe that, with our choice of  $\lambda^*$, we have that {}{$1-2b\rho>0$}. Then, using~\eqref{Eq:Identity}, \arxivonly{we conclude that}
	\begin{IEEEeqnarray}{rCl}
		\Prob\Big[\bigcup_{\hat{\vecn} \in \hat{\Pi}'(\vecn)}\setE(\hat{\vecn}) \Big] &\le&  |\hat{\Pi}'(\vecn)|^\rho e^{-\frac{\n\rho}{2}\log(1-2b \rho)-\n\rho a}.
	\end{IEEEeqnarray}
	From Lemma \ref{Lem:Pi_Cardinality}, it follows that $|\hat{\Pi}'\vecn)| \le \M'$, where $\M'$ is given in \eqref{Eq:M_prime}.  Using the definition of $\tilde{p}_t$, $R$, and $E_0$ in~\mbox{\eqref{Eq:Def_tilde_p}--\eqref{Eq:Def_E0}}, we conclude that $p_t \le  \tilde{p}_t.$
\end{proof}
\section{Numerical Results}
\par In this section, we compute the bound in \eqref{Eq:Final_bound} numerically and compare it with the performance of a {}{TUMA-}adapted version of the CCS-AMP algorithm, originally {}{proposed} in \cite{Amalladinne2020} for {}{the} \gls{UMA} scenario. {}{The original CCS-AMP scheme involves a divide-and-conquer strategy in which messages are split into smaller blocks, as well as the use of inner and outer encoders and decoders}~\cite{fengler2019sparcs}. In particular, the inner decoder reconstructs the transmitted signals, and the outer decoder maps these reconstructions to valid codewords. To extend CCS-AMP {}{to} TUMA, we modify the inner decoder to handle message collisions by incorporating a carefully chosen prior and {}{a} Bayesian estimation of multiplicities. In the outer decoder, we apply a majority-vote approach to estimate {}{the} message multiplicities. \arxivonly{With these adaptations, we are able to use CCS-AMP to estimate multiplicities in TUMA.}

\par {}{In the} numerical evaluation, we {}{set} the codeword length {}{to} $n = \num{38400}$, the number of message bits {}{to} $k = 128$. {}{To use Theorem \ref{Th:Main} we need to specify the message type. We choose it as follows.}  We select  {}{message} types that approximate a {}{Zipf} PMF $\smash{p_Z(m; \n_Z, s) = {m^{-s}}/{\sum_{j=1}^{\n_Z} j^{-s}}}$ for $\smash{m = 1, \hdots, \n_Z}$. To determine the transmitted message type, we use $p_Z(m; \Ma, 1)$. Specifically, we fix a $\Kap$ and calculate each $n_m$ by rounding $\Kap p_Z(m; \Ma, 1)$ so that the total number of users $\smash{\Ka = \sum_m n_m}$ is approximately $\Kap$. The resulting transmitted type is $\vect_{Z} = \vecn / \Ka$. {}{For reference, the UMA multiplicity vector is $\vect_U = [1 / \Kap, \hdots, 1 / \Kap]$.} {When evaluating the bound, we fix $\delta={1}/\lefto({1-2\delta_{\min}^*}\right)-1.01$, and we optimize over $\rho$, with $\rho_{\min}=0.01$.}

\par We evaluate the minimum energy per bit $\smash{E_{\rm b} / N_0 = 10 \log_{10}(\n \power / 2 k)}$ (dB) required to achieve a total variation distance of $\epsilon=0.05$. {In Fig. 1, we depict the value of  $E_{\rm{b}} / N_0$ (dB) versus the total variation distance between the UMA profile and the chosen TUMA profiles,} {}{for} $\Kap = 100$ and $\Ma$ varying from $2$ to $100$. {}{The point $\Ma=100$, which corresponds to the UMA setting, is obtained by using}\mbox{\cite[Theorem 1]{PolyanskiyISIT2017massive_random_access}}. {}{In our implementation of the CCS-AMP scheme for TUMA, we average the total variation distance over 1000 simulations, and we choose the number of potential candidates per sub-blocks to be~$300$.}

\par The trend in $E_{\rm b} / N_0$---which increases, peaks, and then decreases---can be {}{explained} as follows. Assume that the decoder has an estimate of the support of {}{the} transmitted messages. Then, it has to consider $\binom{\Ka-1}{\Ma-1}$ possible multiplicity vectors over that support. This number grows, peaks, and then declines as $\Ma$ decreases from $\Ka'$ to $1$, dictating a similar trend in the  error probability. To maintain a target~$\epsilon$, $E_b/N_0$ must follow the same pattern. Furthermore, as we deviate from the \gls{UMA} profile by {}{decreasing $\Ma$}, the power gain per transmitted message increases, eventually peaking at $\Ma = 1$, where it becomes proportional to ($\Kap)^2$. This power gain accounts} for the eventual  substantial improvement\footnote{Note that when the total variation between the UMA and TUMA profiles is $0.98$, the required $\EbNoi$ (not shown in the figure) is $-27.32$ dB for the bound, and $-22.46$ for CCS-AMP.} in performance, {}{compared to the UMA case when $\Ma$ decreases. We also observe that the the $\EbNoi$ required by CCS-AMP is about $3$dB above the limit predicted by the bound. 
\begin{figure}[htbp]
    \centering
    \begin{tikzpicture}[scale=0.5] 
        \begin{axis}[
            width=1.7\linewidth,
            height=8cm,
            grid=both,
            xlabel={total variation distance between the UMA and TUMA profiles},
            ylabel={$E_{\rm b} / N_0$ (dB)},
            xlabel style={font=\large}, 
            ylabel style={font=\large},
            legend style={at={(0.5,-0.2)}, anchor=north, legend columns=2},
            xmin=0, xmax=1,
            ymin=-9, ymax=5,
            xtick={0, 0.2, 0.4, 0.6, 0.8, 1},
            ytick={-8,-6, -4, -2, 0, 2, 4, 6}, 
            thick ]

        % CCS-AMP data
        \addplot[blue, mark=o, ultra thick] coordinates {
            (0, 3.9955)
            (0.2, 3.7572)
            (0.4, 4.9479)
            (0.5071, 4.2453)
            (0.7, 3.7065)
            (0.8002, -6.89)
            (0.98, -22.4623)
        };

        % Bound in Theorem 1
        \addplot[red, mark=square*, ultra thick] coordinates {
            (0, 0.0452)
            (0.2, 1.51)
            (0.4, 1.44)
            (0.5071, 1.4)
            (0.7, 1.3)
            (0.8002, -7.11)
            (0.98, -27.32)
        };

        % Arrows for curves
        \draw[->, ultra thick] (axis cs:0.4,-4.5) -- (axis cs:0.25,4.3) 
            node[pos=0, below, font=\large] {CCS-AMP};
        
        \draw[->, ultra thick] (axis cs:0.59,-6.15) -- (axis cs:0.45,1.2) 
            node[pos=0, below, font=\large] {Bound (Theorem 1)};

        \end{axis}
    \end{tikzpicture}
    \caption{The required $E_{\rm b}/N_0$ (dB) for message types with $(\Ma,\Ka)$ in $\smash{\{ (100,100), (80,92), (60, 98), (50,102), (30,102), (10,100)\}}$.}
    \label{Fig:EbN0_Bound_Algo}
\end{figure}
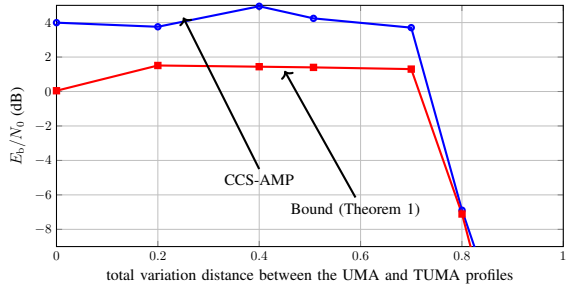

% \begin{figure}
% 	\centering
% 	\includegraphics[width=0.75\linewidth]{EbN0_TV_UMATUMA_14Jan25.png}
% 	\label{Fig:EbN0_Bound_Algo}
% 	\caption{Required $E_b/N_0$ (dB) for different message types, as a function of their
% 			total variation distance from the uniform (UMA) type.}
% 	% \vspace{1mm}
% \end{figure}
\section{Conclusion}
We derived a numerically computable {}{achievability} bound on the {}{error incurred in estimating the type of messages over an unsourced GMAC}. {By evaluating this bound, we obtained insights into the minimum energy per bit required to transmit a given message {}{type} under a total variation constraint.} Extending this work to {}{a} scenario with unknown user and message counts, as well as to fading models, are possible directions for future works.
\arxivonly{
\appendices

\section{Proof of Lemma \ref{Lem:Delta} }
\label{Append:Scaled_TV_Even}
Since $\TV{\vect}{\hT} \in [0,1]$, $\vecnorm{\vecn-\hat{\vecn}}_1 \in[0,2\Ka]$. We have,
\begin{IEEEeqnarray}{rCl}
	\vecnorm{\vecn-\hat{\vecn}}_1&=&\sum_{ m:n_m \ge \hat{n}_m}(n_m-\hat{n}_m)+\sum_{ m:\hat{n}_m > n_m}(\hat{n}_m-n_m).\nonumber\\
	&&
\end{IEEEeqnarray}
Since $\vecnorm{\vecn}_1=\vecnorm{\hat{\vecn}}_1$, we have
\begin{IEEEeqnarray}{rCl}
	\sum_{m: n_m \ge \hat{n}_m}n_m+\sum_{m: \hat{n}_m > n_m}n_m&=&\sum_{ m:n_m \ge \hat{n}_m}\hat{n}_m+\sum_{ m:\hat{n}_m > n_m}\hat{n}_m.\nonumber\\
	&&
\end{IEEEeqnarray}
Rearranging the terms in the above equation, we observe
\begin{IEEEeqnarray}{rCl}
	\sum_{ m:n_m \ge \hat{n}_m}(n_m-\hat{n}_m)&=& \sum_{ m:\hat{n}_m > n_m}(\hat{n}_m-n_m). \nonumber\\
	&&
\end{IEEEeqnarray}
Thus, we observe that $\vecnorm{\vecn-\hat{\vecn}}_1$ is even.
\par Now, we look at the range of values {}{$\vecnorm{\vecn-\hat{\vecn}}_1$} can assume. The upper bound on $\vecnorm{\vecn-\hat{\vecn}}_1$ is achieved by any pair of $\vecn$ and $\hat{\vecn}$ such that $\Su{\vecn}\cap \Su{\hat{\vecn}}=\emptyset$, as
\begin{IEEEeqnarray}{rCl}
	\sum_{ m:n_m \ge \hat{n}_m}(n_m-\hat{n}_m)&=& \sum_{ m:\hat{n}_m > n_m}(\hat{n}_m-n_m)=\Ka.\nonumber\\
	&&
\end{IEEEeqnarray}
Next, since we assume $\vecn \neq \hat{\vecn}$, $\vecnorm{\vecn-\hat{\vecn}}_1 >0$. Since $\vecnorm{\vecn-\hat{\vecn}}_1$ is even, we have $\vecnorm{\vecn-\hat{\vecn}}_1 \ge 2$. The following choices of $\vecn$ and $\hat{\vecn}$ achieves this lower bound. First, fix any $\vecn$ that satisfy the hypothesis of this lemma. Let $m_{\max}$ denote its index with the maximum multiplicity and $m_{\min}$  the index with minimum multiplicity. Choose an $\hat{\vecn}$ such that $\Su{\vecn}= \Su{\hat{\vecn}}$, $\hat{n}_{m_{\max}}=n_{m_{\max}}-1$, $\hat{n}_{m_{\min}}=n_{m_{\min}}+1$, and $\hat{n}_m=n_m$, for all other values of $m\in[\M]$. For such a choice, $\vecnorm{\vecn-\hat{\vecn}}_1=2$. Thus, we have shown that, for any $\vecn$ and $\hat{\vecn}$ satisfying the conditions of the lemma, $\vecnorm{\vecn-\hat{\vecn}}_1 $ is even and belongs to $ [2:2\Ka]$. Alternatively, $\vecnorm{\vecn-\hat{\vecn}}_1=2t$, for some $t\in[\Ka]$.
\arxivonly{
	\section{Proof of Lemma \ref{Lem:Opt_lambda}}
	\label{Sec:Opt_Chernoff}
	Let
	\begin{IEEEeqnarray}{c}
		a'=\frac{1}{2}\log(1+2c\power'\lambda)
	\end{IEEEeqnarray}
	\begin{IEEEeqnarray}{c}
		b'=\lambda\Big(1-\frac{1}{1+2c\power'\lambda}\Big)
	\end{IEEEeqnarray}
	so that
	\begin{IEEEeqnarray}{c}
		E_0'(\lambda)=\rho a'+\frac{1}{2}\log(1-2b'\rho).
	\end{IEEEeqnarray}
	To find the $\lambda$ that maximizes $E_0'$, we analyze the equation $\frac{dE_0'}{d\lambda}=0$ and find the  \textit{stationary point} of the function $E_0'$. That is, we analyze
	\begin{IEEEeqnarray}{c}
		\rho \frac{da'}{d\lambda}=\frac{\rho}{1-2b' \rho}\frac{db'}{d\lambda}.
		\label{Eq:dE0_dlambda}
	\end{IEEEeqnarray}
	We have
	\begin{IEEEeqnarray}{c}
		\frac{da'}{d\lambda}=\frac{c\power' }{1+2c\power'\lambda}
		\label{Eq:da_dlambda}
	\end{IEEEeqnarray}
	and
	\begin{IEEEeqnarray}{c}
		\frac{db'}{d\lambda}=\frac{4c^2\power'^2\lambda^2+4c\power'\lambda}{(1+2c\power'\lambda)^2}.
		\label{Eq:db_dlambda}
	\end{IEEEeqnarray}
	We also have
	\begin{IEEEeqnarray}{c}
		1-2b'\rho=\frac{1+2c\power'\lambda-4c\power'\lambda^2\rho}{1+2c\power'\lambda}.
		\label{Eq:One_min_2b_rho}
	\end{IEEEeqnarray}
	Substituting \eqref{Eq:da_dlambda}, \eqref{Eq:db_dlambda}, and $\eqref{Eq:One_min_2b_rho}$ in \eqref{Eq:dE0_dlambda}, we obtain that the stationary point satisfies
	\begin{IEEEeqnarray}{c}
		c\power'+2c^2\power'^2\lambda-4c^2\power'^2\lambda^2\rho=4c^2\power'^2\lambda^2+4c\power'\lambda.
	\end{IEEEeqnarray}
	Equivalently,
	\begin{IEEEeqnarray}{c}
		4c\power'(1+\rho)\lambda^2+2(2-c\power')\lambda-1=0.
	\end{IEEEeqnarray}
	Solving the quadratic equation, we obtain \eqref{Eq:lambda_star}. From the expression, using the fact that $\rho \in (0,1]$, we can easily observe that
	\begin{IEEEeqnarray}{c}
		\lambda^* \le \frac{1}{2(1+\rho)}< \frac{1}{2}.
	\end{IEEEeqnarray}
    The claim that $\lambda^*$ maximizes $E_0'(\lambda)$ follows  by observing that  $E_0'(\lambda)$ is twice differentiable for $\lambda>0$, and verifying that $\frac{dE_0'(\lambda)}{d\lambda}$ evaluated at $\lambda=\lambda^*$ is negative.
	 }
%Further, we observe that $\lambda^*$ is non-decreasing  in $\rho$, upon fixing other parameters. In such a setting, the maximum value of $\lambda^*$ is achieved (as a function of $\rho$) at $\rho=0$. {}{Since we optimize (numerically) over $\rho\in[\delta',1]$, for some small $\delta'>0$,} there exists a $\delta^*<1/2$ such that $\lambda^*=1/2-\delta^*$, as $\lambda^*$ is a continuous function of $\rho$ (and keeping other parameters fixed) in this domain.
\section{Proof of Lemma \ref{Lem:Exp_f_c}}
\label{Append:Exp_f_c}
We observe that the first derivative of $f(c)$ with respect to $c$ is
\begin{IEEEeqnarray}{c}
	f'(c)=\frac{{\lambda^*} \power'}{(1+2c \power'{\lambda^*} )}\Big(\frac{2{\lambda^*} \vecnorm{\vecz}^2}{1+2c \power'{\lambda^*} }-\n\Big).
	\label{Eq:f_First_Derivative}
\end{IEEEeqnarray}
Further, $f'(c)=0$ is attained at
\begin{IEEEeqnarray}{c}
	c^*=\frac{\vecnorm{\vecz}^2}{\n\power'}-\frac{1}{2\power'{\lambda^*}}.
	\label{Eq:Opt_c}
\end{IEEEeqnarray}
For any $\vecz \le \n(1+\delta)$, i.e., $\vecz \in \setA$, we observe that
\begin{IEEEeqnarray}{c}
	c^* \le \frac{1+\delta}{\power'} -\frac{1}{2\power'{\lambda^*}}.
\end{IEEEeqnarray}
From Lemma \ref{Lem:Opt_lambda},  our choice of $\lambda^*$  satisfies $\lambda^*=1/2-\delta^*$, for some specific $\delta^*\in(0,1/2)$. Further, the  $\delta^*$ so obtained is a function of $c_{\min}$, which in turn depends on the parameters $t$, $\setS$, $\setN$, $\ell$, $i$, and $j$. We choose $\delta^*_{\min}$ as the minimum among all such  $\delta^*$ values. We have $\delta^*_{\min}\in(0,1/2)$ (as we are taking the minimum over a finite set of values).    We choose $\delta$ such that
\begin{IEEEeqnarray}{c}
	\delta < \frac{1}{1-2\delta^*_{\min}}-1.
    \label{Eq:c_star_UB}
\end{IEEEeqnarray}
This will ensure
\begin{IEEEeqnarray}{c}
	\frac{1+\delta}{\power'} -\frac{1}{2\power'{\lambda^*}} <  0.
    \label{Eq:delta_choice}
\end{IEEEeqnarray}
Using $\delta$ as in \eqref{Eq:delta_choice},  with $\vecz\in\setA$, and invoking \eqref{Eq:c_star_UB}, we obtain
\begin{IEEEeqnarray}{c}
	c^* \le 0.
\end{IEEEeqnarray}
It can be easily verified that, for $\vecz \in \setA$,
\begin{IEEEeqnarray}{c}
	c>c^* \Rightarrow \frac{2{\lambda^*} \vecnorm{\vecz}^2}{1+2c \power'{\lambda^*} }<\n.
\end{IEEEeqnarray}
Thus, from the expression for $f'(c)$ in \eqref{Eq:f_First_Derivative}, it follows that,
\begin{IEEEeqnarray}{c}
	c>c^* \Rightarrow f'(c)<0.
\end{IEEEeqnarray}
Since $c^*\le 0$ for $\vecz \in \setA$, $f(c)$ is strictly decreasing in $c$, for $c>0$.

\section{Proof of Lemma \ref{Lem:Pi_Cardinality}}
\label{Append:Cardinality_Pi}
To establish the bound stated in the first claim of the lemma, we need to enumerate all possible solutions $\hat{\vecn}$ of the nonlinear equation $\vecnorm{\vecn-\hat{\vecn}}_1=2t$. We accomplish this by first showing that $\vecnorm{\vecn-\hat{\vecn}}_1$ can be expressed as the sum of two terms which are linear in $\hat{\vecn}$, for every $\hat{\vecn}$ satisfying the assumptions stated in the lemma. Then, we obtain the bound by leveraging the resulting linear structure and estimating the total number of positive and nonnegative integer-valued solutions of an equation of the form~$\smash{x_1+\hdots+x_k=n}$, for~$\smash{k,~n\in\mathbb{N}}$ such that~$k \leq n$.
\par Fix~$\setS$,~$\setN$,~$\widehat{\setN}$, and~$\widehat{\setS}$ as in the statement of the lemma and denote $\setM_{\rm a}=\Su{\vecn}$,  and $\widehat{\setM}_{\rm a}=\Su{\hat{\vecn}}$.  Observe that 
\begin{IEEEeqnarray}{rCl}
 \vecnorm{\vecn-\hat{\vecn}}_1&=&\sum_{\m=1}^{\M}|n_m-\hat{n}_m|\\
 &=&\sum_{m \in \setM_{\rm a}}|n_m-\hat{n}_m|+\sum_{m \in \setM_{\rm a}^c}|n_m-\hat{n}_m|\\
 &=&\sum_{m \in \setS}|n_m-\hat{n}_m|+\sum_{m \in \setM_{\rm a}\cap \widehat{\setM}_{\rm a}}|n_m-\hat{n}_m|\nonumber\\
  &&+\sum_{m \in \widehat{\setS}}|n_m-\hat{n}_m|+\sum_{m \in \setM_{\rm a}^c\cap \widehat{\setM}_{\rm a}^c}|n_m-\hat{n}_m|\\
  &=&\sum_{m \in \setS}n_m+\sum_{m \in  \setM_{\rm a}\cap \widehat{\setM}_{\rm a}}|n_m-\hat{n}_m|+\sum_{m \in \widehat{\setS}}\hat{n}_m+0\nonumber \\
  &&\\
  &=&\vecnorm{\vecn_{\setS}}_1+\sum_{m \in  \setM_{\rm a}\cap \widehat{\setM}_{\rm a}}|n_m-\hat{n}_m|+j\\
   &=&\vecnorm{\vecn_{\setS}}_1+\sum_{m \in  \setM_{\rm a}\cap \widehat{\setM}_{\rm a}\cap \setN}|n_m-\hat{n}_m|+\nonumber\\
   &&+\sum_{m \in  \setM_{\rm a}\cap \widehat{\setM}_{\rm a}\cap \setN^c}|n_m-\hat{n}_m|+j\\
   &=&\vecnorm{\vecn_{\setS}}_1+\sum_{m \in  \setN}\lefto(n_m-\hat{n}_m\right)+\sum_{m \in   \widehat{\setN}}\lefto(n_m-\hat{n}_m\right)+j\nonumber\\
&=&\vecnorm{\vecn_{\setS}}_1+\vecnorm{\vecn_{\setN}}_1-\vecnorm{\hat{\vecn}_{\setN}}_1\nonumber\\
    &&+\vecnorm{\hat{\vecn}_{\widehat{\setN}}}_1-\vecnorm{{\vecn}_{\widehat{\setN}}}_1+j.
\end{IEEEeqnarray}
Next, since $\vecnorm{\vecn-\hat{\vecn}}_1=2t$, we conclude that 
\begin{IEEEeqnarray}{rCl}
\label{Eq:Linear_a}\vecnorm{\vecn_{\setS}}_1+\vecnorm{\vecn_{\setN}}_1-\vecnorm{\hat{\vecn}_{\setN}}_1&=&t \\
\label{Eq:Linear_b}\vecnorm{\hat{\vecn}_{\widehat{\setN}}}_1-\vecnorm{{\vecn}_{\widehat{\setN}}}_1+j&=&t.
\end{IEEEeqnarray}
Here, \eqref{Eq:Linear_a} is linear in $\hat{\vecn}_{\setN}$ and \eqref{Eq:Linear_b} is linear in $\hat{\vecn}_{\widehat{\setN}}$; the two equalities follow by noting that $\smash{\vecnorm{\vecn}_1=\vecnorm{\hat{\vecn}}_1}$,~and~$\smash{\vecnorm{\hat{\vecn}_{\widehat{\setS}}}_1=j}$.~In~\eqref{Eq:Linear_a},~$\setS$~corresponds to the set of messages misdetected at the receiver. Furthermore, $\setN$ represents the set of detected-but-deflated messages, as the estimated multiplicities corresponding to the messages in $\setN$ are deflated with respect to the transmitted multiplicities over $\setN$. Likewise, in~\eqref{Eq:Linear_b},~$\widehat{\setN}$~denotes the set of detected-but-inflated messages by the decoder, as the estimated multiplicities of messages in $\widehat{\setN}$ are strictly inflated compared to the corresponding transmitted multiplicities. Finally,~$\widehat{\setS}$~denotes the set of impostor messages, i.e., messages not transmitted but detected at the receiver. 
\par Next, fix $\setM_{\rm a}\subset [\Ma]$ such that $|\setM_{\rm a}|=\Ma$. Furthermore, fix $\smash{t,~\ell,~i}$ as in the statement of the lemma, fix subsets $\setS$,~$\setN$,~and~$\widehat{\setN}$~of~$\setM_{\rm a}$~such that~$\smash{\setM_{\rm a}=\setS \cup \widehat{\setS} \cup \widehat{\setN}}$, and fix~$\widehat{\setS}\subset \setM_{\rm a}^c\cap[\M]$ such that~$\smash{|\hat{\setS}|=\ell}$,~and~$\smash{\vecnorm{\vecn_{\setS}}_1\ge t}$. Then, every solution~$\hat{\vecn}_{\setN}^*$~ to~\eqref{Eq:Linear_a},~$\hat{\vecn}_{\widehat{\setN}}^*$~\eqref{Eq:Linear_b}, and~$\hat{\vecn}_{\widehat{\setS}}^*$~to~$\smash{\vecnorm{\hat{\vecn}_{\widehat{\setS}}}_1=j}$~define a solution~$\hat{\vecn}_{\setN\cup\widehat{\setN}\cup\widehat{\setS}}^*$~to~$\smash{\vecnorm{\vecn-\hat{\vecn}}_1=2t}$. Accordingly, fixing $t$, $\ell$, $i$, and $j$ as mentioned, $\smash{\setS\subseteq\setM_{\rm a}}$~such that $\smash{|\setS|=\ell}$,~$\smash{\vecnorm{\vecn_{\setS}}_1\ge t}$~$\smash{\setN\subseteq\setM_{\rm a}\cap\setS^c}$ such that $\smash{|\setN|=i}$,~and~$\smash{\widehat{\setN}=\setM_{\rm a}\setminus(\setS \cup \setN)}$. Then
\begin{IEEEeqnarray}{rCl}
\bigcup_{\widehat{\setS}}\bigcup_{\hat{\vecn}_{\widehat{S}}^*}\bigcup_{\hat{\vecn}_{\setN }^*}\bigcup_{\hat{\vecn}_{\widehat{\setN}}^*}\lefto\{\hat{\vecn}^*_{\setN\cup\widehat{\setN}\cup\widehat{\setS}}:\vecnorm{\hat{\vecn}_{\widehat{\setS}}}_1=j,~\text{\eqref{Eq:Linear_a}, and \eqref{Eq:Linear_b} holds}\right\}\nonumber\\
&&\label{Eq:SolutionSet}
\end{IEEEeqnarray}
where $\widehat{\setS}$ is such that $|\widehat{\setS}|=\ell$, contains the set of all solutions to $\smash{\vecnorm{\vecn-\hat{\vecn}}_1=2t}$. In addition, if  $\setN$ forms a deflated set for any $\hat{\vecn}$ of this set (with respect of $\vecn$), i.e., if $\smash{\setN=\{m:n_m \ge \hat{n}_m\}}$,~\eqref{Eq:SolutionSet} represents the set of all solutions to $\smash{\vecnorm{\vecn-\hat{\vecn}}_1=2t}$. 
\par We observe that there are $\binom{\M-\Ma}{\ell}$ ways of choosing an impostor set $\widehat{\setS}$ of length $\ell$ in \eqref{Eq:SolutionSet}, and that
\begin{IEEEeqnarray}{c}
	\binom{\M-\Ma}{\ell} \le \frac{\M^\ell}{\ell!}.
    \label{Eq:Impostor}
\end{IEEEeqnarray}
For every such $\widehat{\setS}$, there are 
\begin{IEEEeqnarray}{c}
    \binom{j-1}{\ell-1}
    \label{Eq:Impost_n}
\end{IEEEeqnarray}
possible $\hat{\vecn}_{\widehat{\setS}}\in\mathbb{N}^{\ell}$ satisfying $\vecnorm{\hat{\vecn}_{\widehat{\setS}}}_1=j$.
\par Next, note that there are $\M^+$ possible $\hat{\vecn}_{\setN}\in \mathbb{N}^i $ satisfying~\eqref{Eq:Linear_a}, where $\M^+$ is defined in \eqref{Eq:M_plus}. In addition, there are  $\M_0^+$ possible~$\vecn_{\setN}-\hat{\vecn}_{\setN}\in\mathbb{N}_0^i$~that satisfies the equation~$\vecnorm{\vecn_{\setN}-\hat{\vecn}_{\setN}}_1=t-\vecnorm{\vecn_{\setS}}_1$,~where $M_0^+$ is defined in \eqref{Eq:M_0_plus}. Hence, there are at most $\min\lefto\{\M_0^+,\M_0\right\}$ valid options for the deflated multiplicities in any solution $\hat{\vecn}$ to $\vecnorm{\vecn-\hat{\vecn}}_1=2t$.
\par Finally, there are
\begin{IEEEeqnarray}{c}
 {t-j-1 \choose \Ma-\ell-i-1 }   
 \label{Eq:InflSol}
\end{IEEEeqnarray}
possible $\hat{\vecn}_{\hat{\setN}}\in \mathbb{N}^{\Ma-\ell-i}$ satisfying \eqref{Eq:Linear_b}. Combining equations~\eqref{Eq:Impostor}--\eqref{Eq:InflSol},~we obtain $\M'$ in \eqref{Eq:M_prime} as the number of solutions of $\hat{\vecn}$ to~$\smash{\vecnorm{\vecn-\hat{\vecn}}_1=2t}$.  
	\par Next, to prove the second part of the lemma, we utilize the inequality
	\begin{IEEEeqnarray}{c}
		\vecnorm{\vecx}_1 \le \sqrt{n}\vecnorm{\vecx}_2,~\vecx \in \reals^n.
        \label{Eq:ell1_ell2_bound}
	\end{IEEEeqnarray}
	We obtain $\vecnorm{\vecn-\hat{\vecn}}^2 \ge c_{\min}$, where $c_{\min}$ is given in \eqref{Eq:c_min_defn},
	by applying \eqref{Eq:ell1_ell2_bound} to $\vecn_{\setS}$, $\vecn_{\setN}-\hat{\vecn}_{\setN}$, $\hat{\vecn}_{\hat{\setN}}-\vecn_{\hat{\setN}}$, and $\hat{\vecn}_{\hat{\setS}}$, separately.
}}

\bibliographystyle{IEEEtran}
\bibliography{IEEEabrv,./references}

\end{document}